\documentclass[sigconf]{acmart}

\usepackage{booktabs}
\usepackage{multirow}
\usepackage{enumitem}
\usepackage[ruled]{algorithm2e}
\usepackage[flushleft]{threeparttable}
\setlist[itemize]{leftmargin=*}

\AtBeginDocument{%
  \providecommand\BibTeX{{%
    \normalfont B\kern-0.5em{\scshape i\kern-0.25em b}\kern-0.8em\TeX}}}

\copyrightyear{2020}
\acmYear{2020}
\setcopyright{acmcopyright}\acmConference[CIKM '20]{Proceedings of the 29th ACM
International Conference on Information and Knowledge Management}{October
19--23, 2020}{Virtual Event, Ireland}
\acmBooktitle{Proceedings of the 29th ACM International  Conference on Information
and Knowledge Management (CIKM '20), October 19--23, 2020, Virtual Event, Ireland}
\acmPrice{15.00}
\acmDOI{10.1145/3340531.3412756}
\acmISBN{978-1-4503-6859-9/20/10}
\acmSubmissionID{ar2285}

\begin{document}
\fancyhead{}

\title{U-rank: Utility-oriented Learning to Rank\\ with Implicit Feedback}

\author{Xinyi Dai$^1$, Jiawei Hou$^1$, Qing Liu$^2$, Yunjia Xi$^1$, Ruiming Tang$^2$, }
\author{Weinan Zhang$^1$,  Xiuqiang He$^2$, Jun Wang$^3$, Yong Yu$^1$}
\affiliation{$^1$Shanghai Jiao Tong University, $^2$Huawei Noah's Ark Lab, $^3$University College London}
\email{{daixinyi, hjw99868, wnzhang, yyu}@sjtu.edu.cn, {liuqing48, tangruiming, hexiuqiang1}@huawei.com, jun.wang@cs.ucl.ac.uk}

\renewcommand{\shortauthors}{X. Dai, et al.}

\begin{abstract}
Learning to rank with implicit feedback is one of the most important tasks in many real-world information systems where the objective is some specific utility, e.g., clicks and revenue. However, we point out that existing methods based on probabilistic ranking principle do not necessarily achieve the highest utility. To this end, we propose a novel ranking framework called \textit{U-rank} that directly optimizes the expected utility of the ranking list.
With a position-aware deep click-through rate prediction model, we address the attention bias considering both query-level and item-level features. Due to the item-specific attention bias modeling, the optimization for expected utility corresponds to a maximum weight matching on the item-position bipartite graph. We base the optimization of this objective in an efficient Lambdaloss framework, which is supported by both theoretical and empirical analysis. 
We conduct extensive experiments for both web search and recommender systems over three benchmark datasets and two proprietary datasets, where the performance gain of \textit{U-rank} over state-of-the-arts is demonstrated. Moreover, our proposed \textit{U-rank} has been deployed on a large-scale commercial recommender and a large improvement over the production baseline has been observed in an online A/B testing.
\vspace{-2pt}
\end{abstract}

\vspace{-4pt}
\vspace{-4pt}

\settopmatter{printacmref=false, printfolios=false}

\maketitle

{\fontsize{8pt}{8pt} \selectfont
\textbf{ACM Reference Format:}\\
Xinyi Dai, Jiawei Hou, Qing Liu, Yunjia Xi, Ruiming Tang, Weinan Zhang, Xiuqiang He, Jun Wang and Yong Yu. 2020. \textit{U-rank}: Utility-oriented Learning to Rank with Implicit Feedback. In \textit{Proceedings of the 29th ACM International Conference on Information and Knowledge Management (CIKM ’20), October19–23, 2020, Virtual Event,  Ireland.} ACM, NY, NY, USA, 8 pages. https://doi.org/10.1145/3340531.3412756}

\section{Introduction}
Ranking is the core of information retrieval.
In the traditional web search scenario, learning to rank (LETOR) methods are proposed to optimize the ranked list based on human-annotated relevance labels~\cite{liu2011learning}. 
Typically these methods sort the documents according to their probability of relevance in descending order, according to the famous probabilistic ranking principle (PRP)~\cite{robertson1977probability}.
Due to the lack of annotated labels, recently, many works have focused on learning to rank via implicit feedback, such as user's click data, which is timely and personalized.
These works are also based on PRP, where the relevance is estimated from implicit feedback through counterfactual methods~\cite{10.1145/2911451.2911537, Joachims2016UnbiasedLW, 10.1145/3159652.3159732,10.1145/3209978.3209986, 10.1145/3308558.3313447}.
Besides the traditional web search scenario, nowadays, ranking is also an important part of many real-world applications, including recommender systems~\cite{karatzoglou2013learning}, online advertising~\cite{tagami2013ctr} and product search~\cite{karmaker2017application}. 
In these applications, specific utility metrics (such as clicks, conversions and revenue, etc.) are proposed, by which the quality of a ranked list is evaluated. 

In many existing works, LETOR algorithms are derived on the basis of PRP, and then evaluated on some utility metrics~\cite{karmaker2017application, zhao2019recommending}. 
However, we find that the PRP based ranking framework does not necessarily bring the highest utility in reality.
To be more specific, PRP is basically correct for items with large differences in relevance estimation. 
However, for two items with close relevance estimation, putting the item which is more sensitive to position change at the higher position will bring a higher expected utility, even if it is slightly less relevant. 
To provide a persuasive example, we show the average click curve of five popular apps from a mainstream App Store in the right panel of Figure~\ref{fig:data_analysis}.
Consider a non-personalized case for simplicity that we recommend App 1, App 2, and App 3 to one user.
If the apps are sorted by PRP, the ranked list will be App 1, App 2, and App 3 by their relevance in terms of click-through rate (CTR). However, the optimal ranked list with the maximum utility should be App 1, App 3, and App 2, since the utility gain of promoting App 3 from the 3rd to 2nd position is 0.019, which is larger than the utility loss 0.010 of dragging App 2 from the 2nd to 3rd position. 
As can be seen, sorting items by relevance may fail to achieve the highest utility in some situations, which is actually quite common in industrial scenarios. Therefore, we aim to optimize the objectives that are directly related to the utility based on user's implicit feedback.

\begin{figure}[t]
\centering
\begin{minipage}[t]{0.175\textwidth}
\centering
\includegraphics[width =\textwidth]{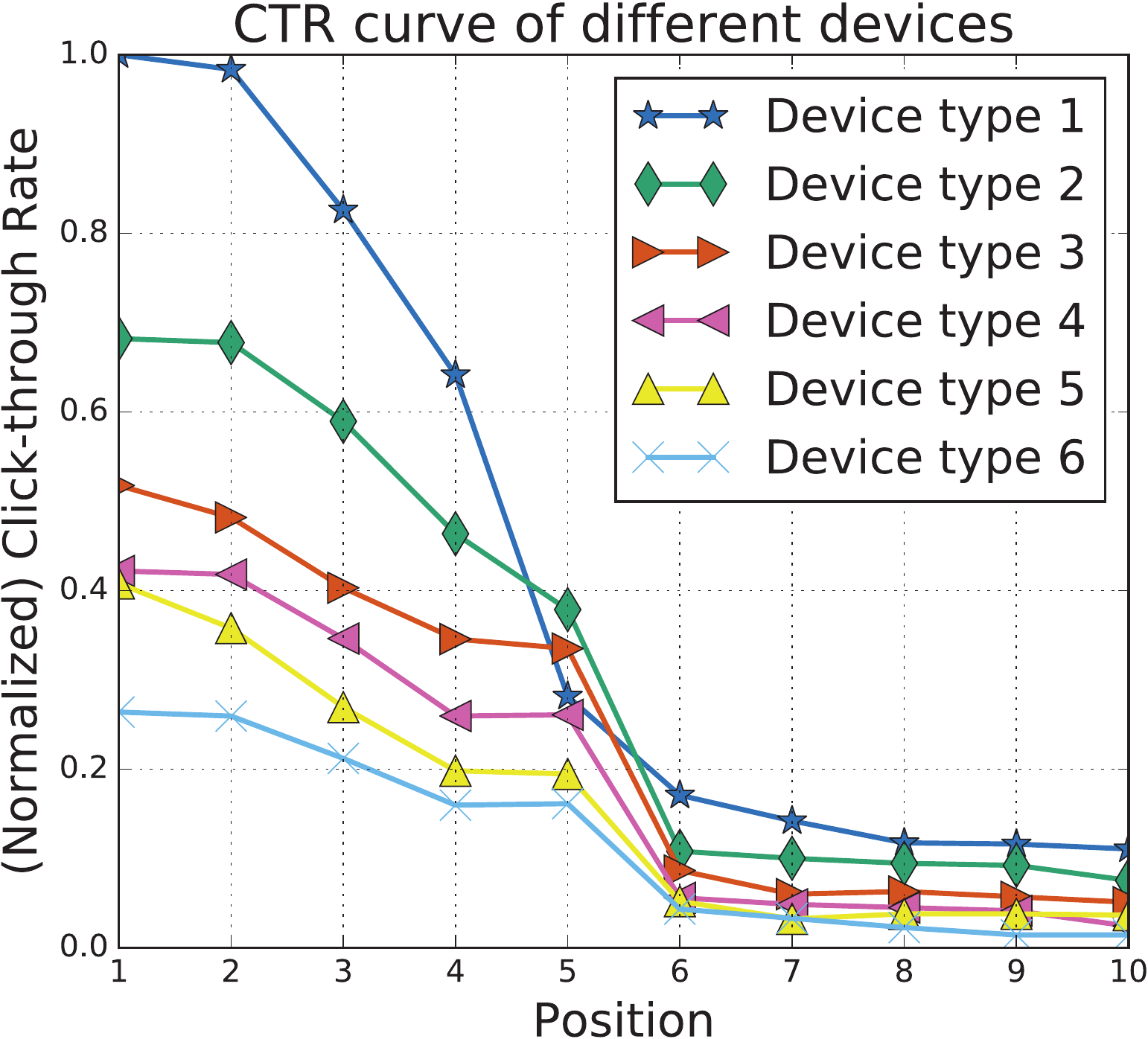}
\label{fig:query-pos}
\end{minipage}
\begin{minipage}[t]{0.295\textwidth}
\centering
\includegraphics[width =\textwidth]{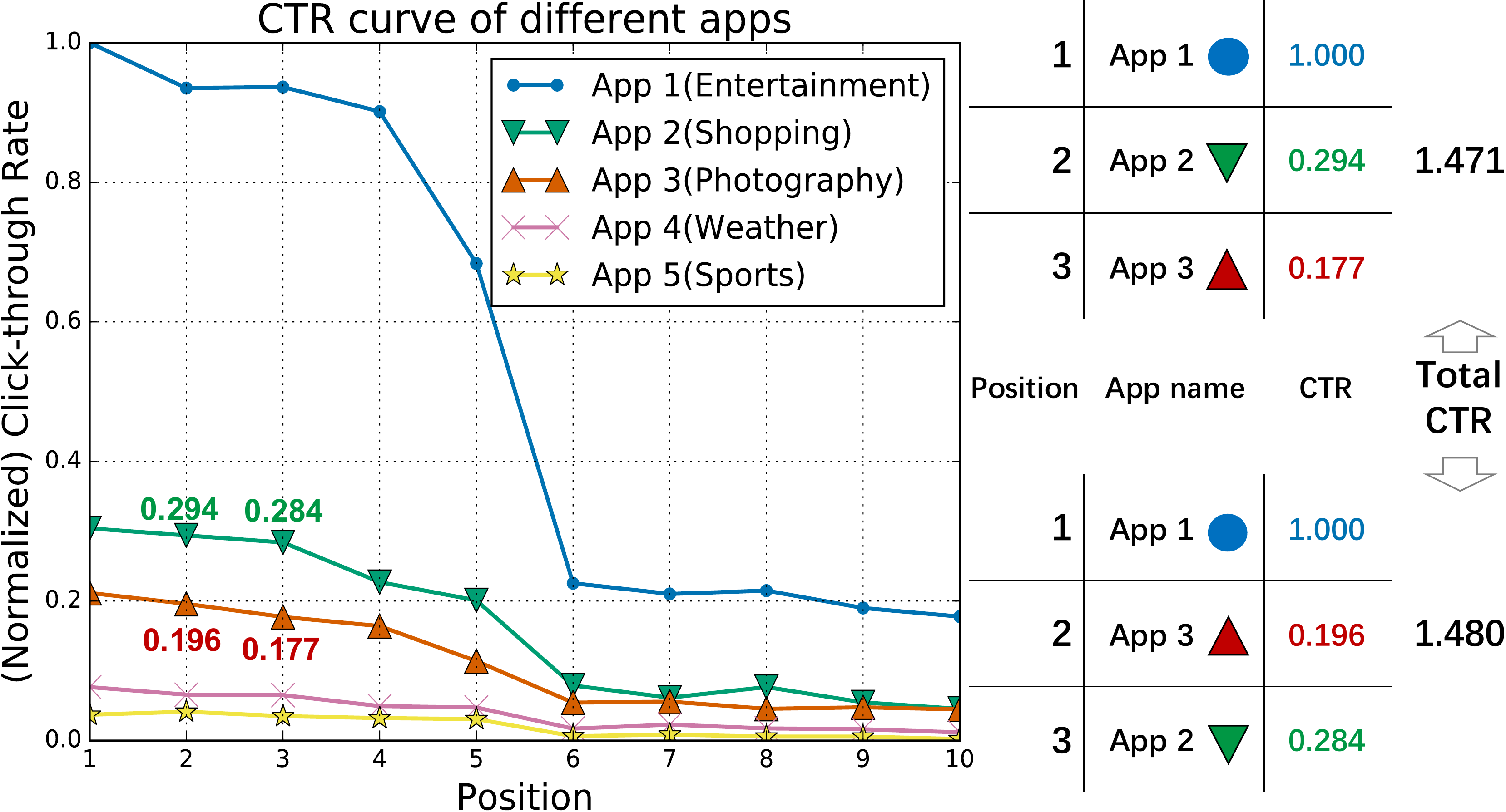}
\label{fig:item-pos}
\end{minipage}
\vspace{-15pt}
\caption{The CTR analysis w.r.t. query/item features. The data is collected through a 120 days' click log on random recommendation traffic in a mainstream App Store.}
\label{fig:data_analysis}
\vspace{-16pt}
\end{figure}

Optimizing the utility metric like CTR of the whole ranked list is not as easy as it seems. 
One direct solution might be estimating \textit{only one unique} 
CTR of each query-item pair and then optimizing the certain metric w.r.t. the whole list using the estimated CTR~\cite{wu2018turning}. 
However, bias will be introduced in this solution since user's CTR is not a static property like relevance. To be clear, for the same query-item pair, the CTR might change with its presented position. As shown in Figure~\ref{fig:data_analysis}, the CTR decreases as the presented position goes from top to bottom, and moreover, the magnitude of the decrease is different among items and device types.

In order to design effective utility-oriented algorithms, we need to figure out why this phenomenon happens and then investigate how to deal with it.
The decrease of user's CTR mainly results from the decrease of user attention, which is supported by eye-tracking studies~\cite{lorigo2008eye,LORIGO20061123}. 
Most existing works have treated such attention bias as position bias~\cite{10.1145/2911451.2911537, Joachims2016UnbiasedLW}, i.e., more attention is paid to the top positions than the bottom ones. In the literature, position bias is considered to be decorrelated with the ranked items, i.e., makes the same effect on all items \cite{PBM,cascade}, which is generally correct in the traditional \textit{10 blue links} scenario. Under such assumption, following PRP achieves the goal of the highest expected utility since the click curves of different positions across different items have the same shape despite the different scales.

However, we argue that in many real-world applications, a user's attention on items does not only depend on the positions but also the item attributes and the user contexts.
For the App recommendation case as demonstrated in Figure~\ref{fig:data_analysis}, visual difference in the thumbnail of a product or the preview frame of videos leads to item-specific attention bias. In web search, an example of item-specific attention bias is vertical bias, commonly observed when the page contains vertical search results (such as images, videos, maps, etc.). For example, \citet{metrikov2014whole} found that an image in search result can raise CTR and flatten the click curve at the same time. A visually attractive content, like a vertical search result or an item with a fancy thumbnail, can still attract user's attention even it is placed at a lower position, leading to a flatter click curve.
In other words, such visually attractive results are less sensitive to the position change. In the example above, CTR of App 2 is less sensitive in whether placing it in position 2 or position 3 compared to App 3. 
Placing items of which CTR is more sensitive to position change at top positions often leads to a higher utility. 
Besides, query-level features like device type, as shown in Figure~\ref{fig:data_analysis}, also leads to different attention biases.
Hence, to obtain unbiased CTR estimation, we need to exploit both the item-level and the query-level features to model the dependency between click and position.

Based on these considerations, in this work, we propose a ranking framework called \textit{U-rank} that directly optimizes expected utility from implicit feedback. Instead of ranking according to PRP, we first derive a new list-wise learning objective, of which the expectation is the utility metric we want to maximize. 
Then to obtain an unbiased estimation of the expected utility, we address the attention bias considering both the query-level and item-level features with a position-aware deep CTR model. Finally, to efficiently optimize the expected utility, we formulate it as an item-position matching problem as shown in Figure~\ref{fig:bigraph}, and learn a scoring function towards the best matching through pairwise permutations inspired by Lambdaloss framework~\cite{lambdaloss}, which reduces the complexity in inference stage from $O(N^3)$ to $O(N)$. Theoretical analysis demonstrates that we solve an upper bound problem of the matching problem.

We conduct thorough experiments on three benchmark LETOR datasets and a large-scale real-world commercial recommendation dataset to verify the effectiveness of \textit{U-rank}. Further, \textit{U-rank} has been deployed on the recommender system of a mainstream App Store, where a 10-day online A/B test shows that \textit{U-rank} achieves an average improvement of 19.2\% on CTR and 20.8\% on conversion rate over the production baseline.

\section{Related Work}
Generative click models are introduced to study user browsing behavior and extract unbiased relevance feedbacks from click data. For example, Position-based model (PBM)~\cite{PBM} assumes that a click only depends on the position and the relevance of the document. Cascade model~\cite{cascade} assumes that user browses a search web page sequentially from top to bottom until a relevant document is found. Following these two classical click models, more sophisticated ones (e.g., UBM~\cite{10.1145/1390334.1390392}, DBN~\cite{10.1145/1526709.1526711}, CCM~\cite{10.1145/1526709.1526712} and NCM~\cite{10.1145/2872427.2883033}) have been proposed. 
These click models estimate the relevance of each item in a point-wise manner instead of considering the relative order of the items as in pairwise or listwise approaches.
Recently, a new line of research, referred to as counterfactual methods, utilizes inverse propensity score (IPS) weighting to address position bias in a learning to rank framework. \citet{10.1145/2911451.2911537} and \citet{Joachims2016UnbiasedLW} proposed the IPW-based framework of debiasing click data in a learning to rank framework. In both works, the propensity estimation relies on randomizing search results displayed to users, which obviously degrades users' search experience. Considering this, \citet{10.1145/3289600.3291017} proposed PBM to estimate propensity without Intrusive Interventions. CPBM ~\cite{Fang2018InterventionHF}, on the basis of PBM, learns a query-dependent propensity estimation. However, multiple rankers are required to learn, which makes them inconvenient to deploy in real-world applications. Besides, another branch of unbiased learning to rank works~\cite{10.1145/3159652.3159732,10.1145/3209978.3209986, 10.1145/3308558.3313447} jointly learn the propensity model with a relevance model, which results in biased estimation of propensity unless the relevance estimation is very accurate.

 \section{Problem Formulation}\label{section:preliminary}
\begin{figure}[t]
	\centering
	\includegraphics[width=0.45\textwidth]{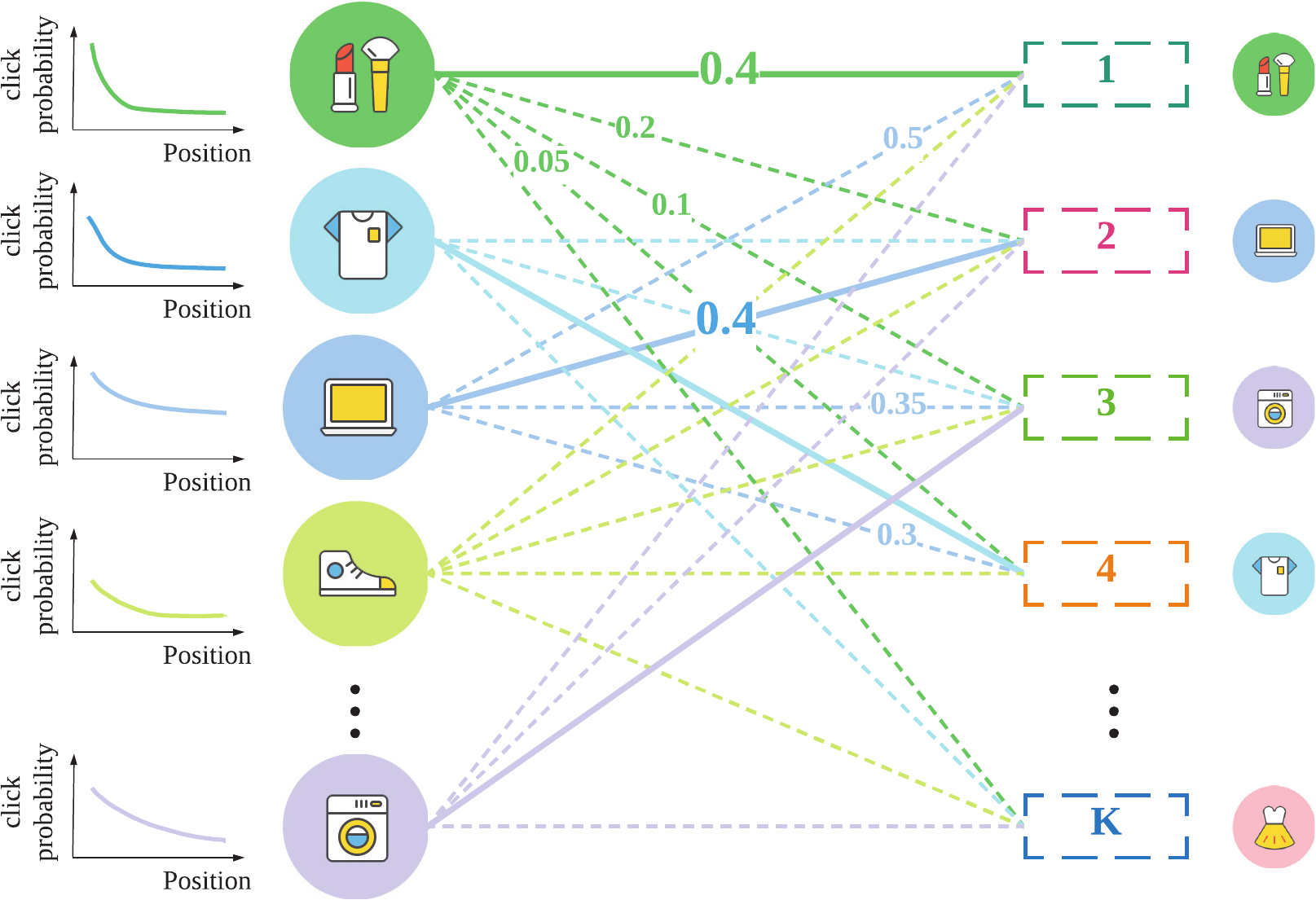}
	\caption{The maximization of the utility can be seen as solving the maximum-weight matching on the item-position bipartite graph, where the edge weight between an item and a position denotes the utility of placing the item at this position, i.e., the product of item's CTR at this position and the utility value of this item.}
	\vspace{-20pt}
	\label{fig:bigraph}
\end{figure}

When a user issues a new request $q$, the system delivers a ranked list $(f_i, b_i)_{i=1}^{n_q}$ of $n_q$ items to the user according to a ranking model
over all the candidate items. 
The feature vector $f_i$ of each item $i$ consists of item features, context features, and user/query features.
The scalar $b_i$ denotes the utility value related to each item $i$, \textit{e.g.}, the watch time of each video in video recommendation, or the bid price of each ad in sponsored search.

The users' click logs are a set $S = \{(f_i, k_i, b_i, c_{i,k_i})_{i=1}^{n_q}\}_{q\in Q}$, where $k_i$ is the position of item $i$, and $c_{i,k_i}$ is the users' implicit feedback on item $i$ when displaying at position $k_i$, \textit{i.e.}, $c_{i,k_i}=1$ for click and $c_{i,k_i}=0$ for non-click. To distinguish between the position of item $i$ in users' click logs and in the current ranking model, in the following parts, we use $k_i^{h}$ to denote the position of item $i$ in click logs and keep $k_i$ as the position of item $i$ in the current ranking model. 

The ultimate goal of this system is to find the best permutation of candidate items for each query $q$ to maximize the utility. The utility is defined as the expected sum of weighted clicks of each item over the whole ranked list, as follows,
\begin{equation}
\begin{aligned}
     U_q = \mathbb{E} \Big[\sum_{i=1}^{n_q} \text{c}_{i, k_{i}} \cdot b_i \Big]
     = \sum_{i=1}^{n_q} P(\text{c}_{i, k_{i}} = 1) \cdot b_i ~,
\end{aligned}
\label{eq:utility}
\end{equation} 
where $P(\text{c}_{i, k_{i}} = 1)$ is the probability of the item $i$ being clicked if displayed at position $k_i$. Maximizing utility $U_q$ is equivalent to solving a maximum weight matching problem on the item-position bipartite graph, where $P(\text{c}_{i, k_{i}} = 1) \cdot b_i$ is the edge weight between the item $i$ and the position $k_i$ in the graph, as shown in Figure~\ref{fig:bigraph}.

\section{Model Framework}
In this section, we present a general ranking framework, i.e., \textit{U-rank}, to maximize the utility $U_q$ in Eq.~\eqref{eq:utility} directly.
Firstly, we derive an unbiased metric of utility from click logs, the expectation of which is the utility $U_q$. Secondly, we design an efficient learning to rank method to optimize this metric, of which the loss function is an upper bound of the utility regret. 

\subsection{Unbiased Estimation of the Utility}
\label{sec:click}
The main difficulty of existing methods of learning to rank via implicit feedback lies in the estimation of the underlying attention bias (or position bias), since we do not observe them directly from the data. With the new learning objective, we do not need to infer relevance or the attention bias explicitly. Instead, we have to deal with another mismatch problem, which is between the 
CTR of the historical presented position and that of the final presented position. For example, if one item is ranked first in the click logs but presented at the 10th position in the final ranking, then its utility is overestimated. To correct this bias, we need an accurate model of user's CTR on different positions. 

The estimation of CTR on different position $P(c_{i,k_{i} }=1)$ refers to one of the most well-studied tasks in recommender systems, i.e., CTR estimation. Deep CTR models~\cite{zhang2016deep,qu2018product} can take position and the rich query-item features as input, to model the complex user interaction in feature space from the click logs.
It is pointed out that position $k_i$ is a very important feature in CTR estimation~\cite{bai2019position,guo2019pal}. However, if we directly used a CTR estimation model as the ranking model, the position feature is vacant at the inference stage. 
Therefore, we design a position-aware deep CTR model as a debiasing module instead of directly using it as a ranking model.
Assume that the probability function of item $i$ displayed at position $k_i$ is a function $g_\theta$ of item feature $f_i$ and position $k_i$, \text{i.e.}, $P(c_{i,k_{i} }=1) = g_\theta(f_i, k_i)$. Then we can estimate the parameter $\theta$ via the standard cross-entropy minimization:  
\begin{equation}
\mathcal{L}_c(\theta) = \sum_{q \in Q} \sum_{i=1}^{n_q} l\left(c_{i, k_i}, g_\theta \left(f_{i}, k_i \right)  \right)~,
\label{eq:ctrloss}
\end{equation}
where $l(p,q) = - p \log q - (1-p)\log(1-q)$ is the cross-entropy loss.

Based on users' click logs and the estimated CTR, we derive an unbiased metric of utility $U_q$ as
\begin{equation}
\begin{aligned}
    U'_q & = \sum_{i=1}^{n_q}  c_{i,k_{i}^{h} } \cdot \frac{P(c_{i,k_{i} }=1)} {P(c_{i,k_{i}^{h} }=1)} \cdot b_i ~. 
\end{aligned}
\label{eq:origin_metric}
\end{equation}

We prove that our derived utility $U'_q$ is unbiased w.r.t. $U_q$ in Eq.~\eqref{eq:utility}, by showing that the expectation of $U'_q$ is equivalent to $U_q$, as
\begin{equation}
\begin{aligned}
    \mathbb{E}[U'_q]
    & = \mathbb{E}\Big[\sum_{i=1}^{n_q}  c_{i,k_{i}^{h} } \cdot \frac{P(c_{i,k_{i} }=1)} {P(c_{i,k_{i}^{h} }=1)} \cdot b_i \Big]\\
    & = \sum_{i=1}^{n_q}  P(c_{i,k_{i}^{h} }=1) \cdot \frac{P(c_{i,k_{i} }=1)} {P(c_{i,k_{i}^{h} }=1)} \cdot b_i \\ 
    & = \sum_{i=1}^{n_q} P(\text{c}_{i, k_{i} } = 1) \cdot b_i = U_q ~.
\end{aligned}
\end{equation}

\subsection{Learning to Optimize the Utility}
One straightforward way to optimize $U_q$ is to perform maximum-weight matching algorithm (e.g., Kuhn-Munkres algorithm~\cite{kuhn1955hungarian,munkres1957algorithms}) on the bipartite graph directly (each query corresponds to one graph), given $g_\theta(f_i, k_i)$. However, the complexity for such a graph matching algorithm to run in the inference stage is $O(N^{3})$ ($N$ denotes the number of candidate items), which is unacceptable in a production system. 
Therefore, in this section, we propose a parameterized scoring function $\Phi(\cdot)$ to approximate the maximum-weight matching procedure on each query, still aiming at maximizing the utility, so that the complexity at the inference stage can be reduced to $O(N)$. 
For each item $i$, the scoring function $\Phi(\cdot)$ gives a ranking score as $s_i=\Phi(f_{i}, b_i)$.
For each query $q$, we compute the score $s_i$ of each item $i$, and the result list is generated by sorting their scores in descending order. 

According to Eq.~\eqref{eq:origin_metric}, we define the utility of displaying item $i$ at position $k_i$ as $u(i, k_i) = c_{i,k_{i}^{h} } \cdot \frac{P(c_{i,k_{i} }=1)} {P(c_{i,k_{i}^{h} }=1)} \cdot b_i $.
With $k_i^*$ being the optimal position assigned to item $i$ by the graph matching algorithm, the regret of the utility is defined as
\begin{equation}
\begin{aligned}
    \mathcal{L}_r(\Phi , q) =  \sum_{i=1}^{n_q} u(i, k_i^*) - \sum_{i=1}^{n_q} u(i, k_i)~.
\end{aligned}
\end{equation}
Minimizing the regret of the utility $\mathcal{L}_r(\Phi,q)$ directly is infeasible as $k_i$'s are discrete values.
Therefore, we adapt the LambdaLoss framework~\cite{lambdaloss} to learn a ranking model towards the optimal ranking by optimizing our proposed loss function (which will be presented in Eq.~(\ref{eq:our-lambda-obj})) with iterative pairwise permutation. Like in LambdaLoss we follow an EM procedure that in E step we obtain the ranked list based on current scoring function $\Phi^{(t)}$ and in M step we re-estimate the scoring function $\Phi^{(t+1)}$ to minimize our loss function. The learning procedure of \textit{U-rank} is as follows. 

We first initialize the scoring function with random initialization of $\Phi^{(0)}$. 
Inspired by the re-weighting technique used in LambdaRank~\cite{burges2007learning}, we compute the difference between the unbiased utility $\Delta Util(i, j)$ when the positions of two items $i$ and $j$ are swapped, as
\begin{equation}
\begin{aligned}
    \Delta Util(i, j) = u(i, k_j) + u(j, k_i) - u(i, k_i) - u(j, k_j) ~.
\label{eq:delta_rev}
\end{aligned}
\end{equation}
Then this difference value is used as the weight in the pairwise loss for each pair of items. 
Following~\cite{burges2005learning,burges2007learning}, we design our loss function in the form of logistic loss, as 
\begin{equation}
\begin{aligned} \label{eq:our-lambda-obj}
    \mathcal{L}_r^{\prime}(\Phi, q) & =  \sum_{i=1}^{n_q} \sum_{j : k_{j} <k_{i}}  \Delta Util(i, j) \log \left(1+e^{-\sigma\left(s_{i}-s_{j}\right)}\right), \\
\end{aligned}
\end{equation}
where $k_i$ and $k_j$ denote the position assigned to item $i$ and $j$ by ranking model at the last step (by the scoring function $\Phi^{(t)}$). This loss is minimized, so that we get a new scoring function $\Phi^{(t+1)}$. 
Then  the process is repeated until convergence. 

Notice that in a standard LambdaLoss framework, the LambdaLoss is defined as
\begin{equation}\label{eq:lambdaloss-obj}
    \mathcal{L_\lambda}(\Phi, q) =  \sum_{i=1}^{n_q} \sum_{j : y_{i} > y_{j}}  |\Delta NDCG(i, j)| \log \left(1+e^{-\sigma\left(s_{i}-s_{j}\right)}\right) ~.
\end{equation}

Note that the differences between our objective (\ref{eq:our-lambda-obj}) and the LambdaLoss objective (\ref{eq:lambdaloss-obj}) lie in (i) the subscript of the summation symbol and (ii) the absolute value symbol of the difference term $\Delta$. In LambdaLoss framework, the pairwise label of each item pair $(i,j)$ is determined. The optimal ranking order is known to us by ranking all the items according to relevance label or click label (denoted by $y_i$ for item $i$), in descending order. 
However, in our framework, we cannot obtain the explicit label $y_i$ for item $i$ . 
An item is treated as the positive item if it is placed at a lower position by scoring function $\Phi^{(t)}$ and the swap of the item pair brings utility gain, and vise versa. We cannot access the optimal ranking order in each query beforehand, where the optimal order is achieved through iterative pairwise permutation. 

\vspace{-5pt}
\subsection{Theoretical Analysis}
In this section, we theoretically prove that our proposed training objective $\mathcal{L}_r^{\prime}(\Phi, q)$ is an upper bound of the utility regret $\mathcal{L}_r(\Phi , q)$. To make the proof easier to understand, we construct a function:  
\begin{equation}
\mathcal{L}_r^{\prime \prime}(\Phi, q) = \sum_{i=1}^{n_q} \sum_{j:s_i<s_j} |u(i, k_j)- u(i, k_i)|.
\end{equation}
We start with several lemmas which will be used in our proof.  

\begin{lemma}
Given an indicator function $f(x)=\mathbb{I}(x \le 0)$ and a function $g(x) = \log (1+e^{-\sigma x})$ where $\sigma$ is a constant in $\mathbb{R}$, it holds that $f(x) \le g(x)$ for all $x \in \mathbb{R}$. 
\label{lem:1}
\end{lemma}

\begin{lemma}
Given an indicator function $f(x)=\mathbb{I}(x \ge 0)$ and a function $g(x) = \max\{\log(1+e^{\sigma C}), 2\} - \log (1+e^{-\sigma x})$ where $\sigma$ is a constant in $\mathbb{R}$, it holds that $f(x) \le g(x)$ for $x\in [-C, C]$.
\label{lem:2}
\end{lemma}

\begin{lemma}
Given a sum function $ f(x) = \sum_i x_i $ and a max function $g(x)= \max_i x_i$, it holds that $f(x) \ge g(x)$ for $x_i \ge 0, \forall i$.
\label{lem:3}
\end{lemma}

Now we are ready to derive the main theoretical result. 
\vspace{-3pt}
\begin{theorem}

Assume $u(i, k_i)$ is a monotonic decreasing function w.r.t $k_i$ and the ranking score $s_i$ is bounded in the range of [-C,C]. Let $C_1 = max\{log(1 + exp(2\sigma C)),2\}$ and $C_2= C1 \cdot \sum_{j=1}^{n_q} \sum_{i:k_i>k_j} (u(j, k_j) - u(j, k_i)) $. Then we have 
$\mathcal{L}_r^{\prime \prime}(\Phi, q) \le \mathcal{L}_r^{\prime}(\Phi, q) + C_2 $.
\label{theorem:1}
\end{theorem}
\begin{proof}
{\small
\begin{align}
\mathcal{L}_r^{\prime \prime} (\Phi, q)  =& \sum_{i=1}^{n_q} \sum_{j=1}^{n_q} |u(i, k_j)- u(i, k_i)| \mathbb{I} (s_i<s_j) \nonumber\\
 =& \sum_{i=1}^{n_q} \sum_{j:k_j<k_i} (u(i, k_j)- u(i, k_i)) \mathbb{I} (s_i<s_j)  + \sum_{j=1}^{n_q} \sum_{i:k_i>k_j} (u(j, k_j)\nonumber\\
 & - u(j, k_i)) \mathbb{I} (s_j<s_i)\nonumber\\
 \le & \sum_{i=1}^{n_q} \sum_{j:k_j<k_i} (u(i, k_j)- u(i, k_i)) \log \left(1+e^{-\sigma\left(s_{i}-s_{j}\right)}\right) \nonumber \\
 &  - \sum_{i=1}^{n_q} \sum_{j:k_j<k_i} (u(j, k_j)  - u(j, k_i)) [\log \left(1+e^{-\sigma\left(s_{i}-s_{j}\right)}\right) + C_1] \nonumber\\
 =& \sum_{i=1}^{n_q} \sum_{j:k_j<k_i} \Delta Util(i,j) \log \left(1+e^{-\sigma\left(s_{i}-s_{j}\right)}\right) + C_2 \\
 =&\mathcal{L}_r^{\prime}(\Phi, q) + C_2 ~\label{eq:derive}
\end{align}
}
where the inequality holds due to Lemma \ref{lem:1} and Lemma \ref{lem:2}.
\end{proof}

Theorem \ref{theorem:1} states that $\mathcal{L}_r^{\prime \prime} (\Phi, q)$ is upper bounded by our objective $\mathcal{L}_r^{\prime} (\Phi, q)$ plus $C_2$. $C_2$ is a constant since in the M step $C_2$ only depends on the current scoring function $\Phi^{(t)}$. Notice that the assumptions in the theorem are not restrictive in practice. As illustrated in Figure \ref{fig:data_analysis}, the real utility basically satisfies the monotonic decreasing assumption. Moreover, the ranking score is often clipped in implementation to avoid explosion in exponential function.

\begin{theorem}
Assume the utility $u(i, k_i)$ is a monotonic decreasing function w.r.t $k_i$. Then $\mathcal{L}_r (\Phi, q)$ is upper bounded by $\mathcal{L}_r^{\prime \prime} (\Phi, q)$.
\label{theorem:2}
\end{theorem}
\vspace{-10pt}
\begin{proof}
{\small
\begin{align}
\vspace{-10pt}
\mathcal{L}_r^{\prime \prime} (\Phi, q) &= \sum_{i=1}^{n_q} \sum_{j:s_i<s_j} |u(i, k_j)- u(i, k_i)|  
= \sum_{i=1}^{n_q} \sum_{k=1}^{k_i-1} |u(i, k)- u(i, k_i)| \nonumber\\
& = \sum_{i=1}^{n_q} \sum_{k=1}^{k_i-1} (u(i, k)- u(i, k_i)) 
\ge \sum_{i=1}^{n_q} u(i, 1) - \sum_{i=1}^{n_q} u(i, k_i)\nonumber\\
&\ge \sum_{i=1}^{n_q} u(i, k_i^*) - \sum_{i=1}^{n_q} u(i, k_i)= \mathcal{L}_r (\Phi, q)
\end{align}
}
where the first inequality holds due to Lemma \ref{lem:3}.
\end{proof}

\begin{table*}[htbp]
	\centering
	\small
	\caption{Comparison of different unbiased learning to rank models on three benchmark datasets. }
	\vspace{-9pt}
	\scalebox{0.9}{
	\begin{tabular}{ c | c || c | c | c | c | c | c | c | c |c|c|c|c } 
		\hline
		\multicolumn{2}{c||}{\multirow{2}{*}{Ranking model}} & \multicolumn{4}{c|}{Yahoo! LETOR set 1} & \multicolumn{4}{c|}{MSLR-WEB10K} & \multicolumn{4}{c}{Istella-S LETOR} \\ \cline{3-14} 
		\multicolumn{2}{c||}{} & MAP & nDCG  & \# Click &CTR & MAP & nDCG& \# Click & CTR & MAP & nDCG & \# Click & CTR \\\cline{1-14}
		\multirow{4}{*}{SVMRank}
		& \multicolumn{1}{c||}{None} & 0.702  & 0.845 & 0.599 & 0.0641 & 0.498 & 0.735 & 0.834 & 0.0827  & 0.773 &  0.808 & 0.931 & 0.0939\\\cline{2-2}
		& \multicolumn{1}{c||}{Randomization} & 0.639  & 0.787 & 0.544 & 0.0574 & 0.433 & 0.686 & 0.820 & 0.0799 & 0.742 & 0.787 & 0.909 & 0.0910 \\\cline{2-2}
 		& \multicolumn{1}{c||}{CPBM} & 0.701  & 0.843 & 0.594 & 0.0637  & 0.477 & 0.721 & 0.751 & 0.0746 & 0.752 & 0.793 & 0.923 & 0.0919 \\ \cline{2-2}
 		& \multicolumn{1}{c||}{\textit{Groundtruth}} & \textit{0.718} & \textit{0.859} & \textit{0.612} & \textit{0.0656} & \textit{0.515} & \textit{0.748}  & \textit{0.872} & \textit{0.0870} & \textit{0.775} & \textit{0.816} & \textit{0.958} & \textit{0.0952}\\  \hline
 		\multirow{4}{*}{LambdaRank}
 		& \multicolumn{1}{c||}{None} & 0.700  & 0.847 & 0.606 & 0.0641 & 0.498 & 0.736 & 0.834 & 0.0828& 0.776  & 0.810 & 0.945 & 0.0941\\ \cline{2-2}
 		& \multicolumn{1}{c||}{Randomization} & 0.680  & 0.828 & 0.582 & 0.0621 & 0.451 & 0.700 & 0.813 & 0.0808  & 0.748 & 0.793 & 0.924 & 0.0923 \\ \cline{2-2}
 		& \multicolumn{1}{c||}{CPBM} & 0.718  & 0.857 & 0.613 & 0.0651 & \textbf{0.514} &\textbf{0.744} & 0.836 & 0.0836 & 0.779 & 0.813 & 0.941 & 0.0941 \\ \cline{2-2}
 		& \multicolumn{1}{c||}{\textit{Groundtruth}} & \textit{0.719} & \textit{0.859} & \textit{0.618} & \textit{0.0657}  & \textit{0.521} & \textit{0.748} & \textit{0.882} & \textit{0.0883} & \textit{0.781} & \textit{0.815} & \textit{0.948} & \textit{0.0948} \\ \hline
 		
 		\multirow{3}{*}{DNN}
 		& \multicolumn{1}{c||}{DLA} & 0.639 & 0.782 & 0.553 & 0.0589  & 0.430 & 0.682 & 0.830 & 0.0839 & 0.676 & 0.703 & 0.828 & 0.0824\\ \cline{2-2}
		& \multicolumn{1}{c||}{CTR-1} & 0.647 & 0.792 & 0.551 & 0.0577  & 0.477 & 0.722 & 0.829 & 0.0814 & 0.733  & 0.771 & 0.894 & 0.0895  \\ \hline
	    \multicolumn{2}{c||}{U-rank} & \textbf{0.719} & \textbf{0.861*}& \textbf{0.618*} & \textbf{0.0659*} & 0.492 & 0.725 & \textbf{0.903*} & \textbf{0.0915*} & \textbf{0.783*} & \textbf{0.816*} & \textbf{0.959*} & \textbf{0.0953*}\\ \hline
		\multicolumn{2}{c||}{\textit{KM (oracle model)}} & \textit{0.935} & \textit{0.987} & \textit{0.684} & \textit{0.0737}  & \textit{ 0.710} & \textit{0.723} & \textit{0.976} & \textit{0.0969} & \textit{0.993} & \textit{0.995} & \textit{1.132} & \textit{0.1126}\\ \hline
	\end{tabular}
	}
    \footnotesize \flushleft $*$ denotes statistically significant improvement (measured by t-test with p-value$<$0.05) over all baselines. Note: baselines in italic is not included.
    \vspace{-10pt}
	\label{tab:simulation}
\end{table*}

\begin{theorem}
Under the assumption of Theorem \ref{theorem:1} and Theorem \ref{theorem:2}, we have that $\mathcal{L}_r(\Phi, q) \le \mathcal{L}_r^{\prime}(\Phi, q) + C_2 $.
\label{theorem:3}
\end{theorem}

The proof of Theorem~\ref{theorem:3} is trivial due to Theorem~\ref{theorem:1} and Theorem~\ref{theorem:2}. Theorem \ref{theorem:3} demonstrates that the utility regret $\mathcal{L}_r(\Phi, q)$ is bounded by our proposed objective $\mathcal{L}_r^{\prime}(\Phi, q)$ plus a constant. It implies that optimizing our proposed objective is actually minimizing the upper bound of the utility regret, which guarantees the effectiveness of \textit{U-rank} theoretically.

\section{\hbox{Semi-synthetic Experiments}}
The semi-synthetic setup is widely applied in unbiased learning to rank~\cite{Fang2018InterventionHF} which allows us to explore different settings~\footnote{Code for our experiments is available at https://github.com/xydaisjtu/U-rank}.

\subsection{Datasets}
\begin{itemize}
    \item \textbf{Yahoo! LETOR set 1}\footnote{https://webscope.sandbox.yahoo.com} is used in Yahoo! Learning-to-Rank Challenge. 
    The dataset consists of 700 features normalized in $[0,1]$, which are extracted from query-document pairs.
    
    \item \textbf{MSLR-WEB10K}\footnote{https://www.microsoft.com/en-us/research/project/mslr/} is a large-scale dataset released by Microsoft Research. It contains 10,000 queries and 1,200,193 documents. There are 136 features extracted from query-document pairs. 
    
    \item \textbf{Istella-S LETOR}\footnote{http://quickrank.isti.cnr.it/istella-dataset/}~\cite{Lucchese2016PostLearningOO} is one of the largest public available datasets. Istella-S is composed of 33,018 queries, where for each query-document pair there are 220 features.  
\end{itemize}

\subsection{Click Data Generation}
We mainly follow \citet{Fang2018InterventionHF} to generate synthetic click data with item-specific attention bias for the three datasets.
In the following part, \textit{oracle model} refers to this click generation model.
Similar to~\cite{Fang2018InterventionHF}, the attention bias which is related to both position and the item feature is calculated by $P(o_{i}=1|k_i, x_i) = 1/k_i^{\max(w^\top x_i + 1, 0)}$. In our setting, $x_i$ is the set of item features, while in the setting of~\cite{Fang2018InterventionHF}, $x_i$ is the set of query features which is shared among all the items for a same query. The parameter vector $w$ is drawn from a uniform distribution over $[-\eta, \eta)$ and is normalized such that $\sum_{j=1}w_j = 0$.
Following \citet{10.1145/3308558.3313447}, the relevance probability is defined as $P\left(r_{i} = 1\right)=\epsilon+(1-\epsilon) \frac{2^{y_i}-1}{2^{y_{\max }}-1}$, where $y_i$ denotes the relevance label of $x_i$ and $y_{max}$ is the highest level of relevance. $\epsilon$ is set to 0.1, which denotes the CTR of irrelevant documents. The overall CTR is calculated by $P(c_i=1) = P(o_i = 1) \cdot P(r_i = 1)$. The maximal position $k_{max}$ is set to be 10.

\begin{figure*}[hbt]
    \centering
    \includegraphics[width=0.3\textwidth]{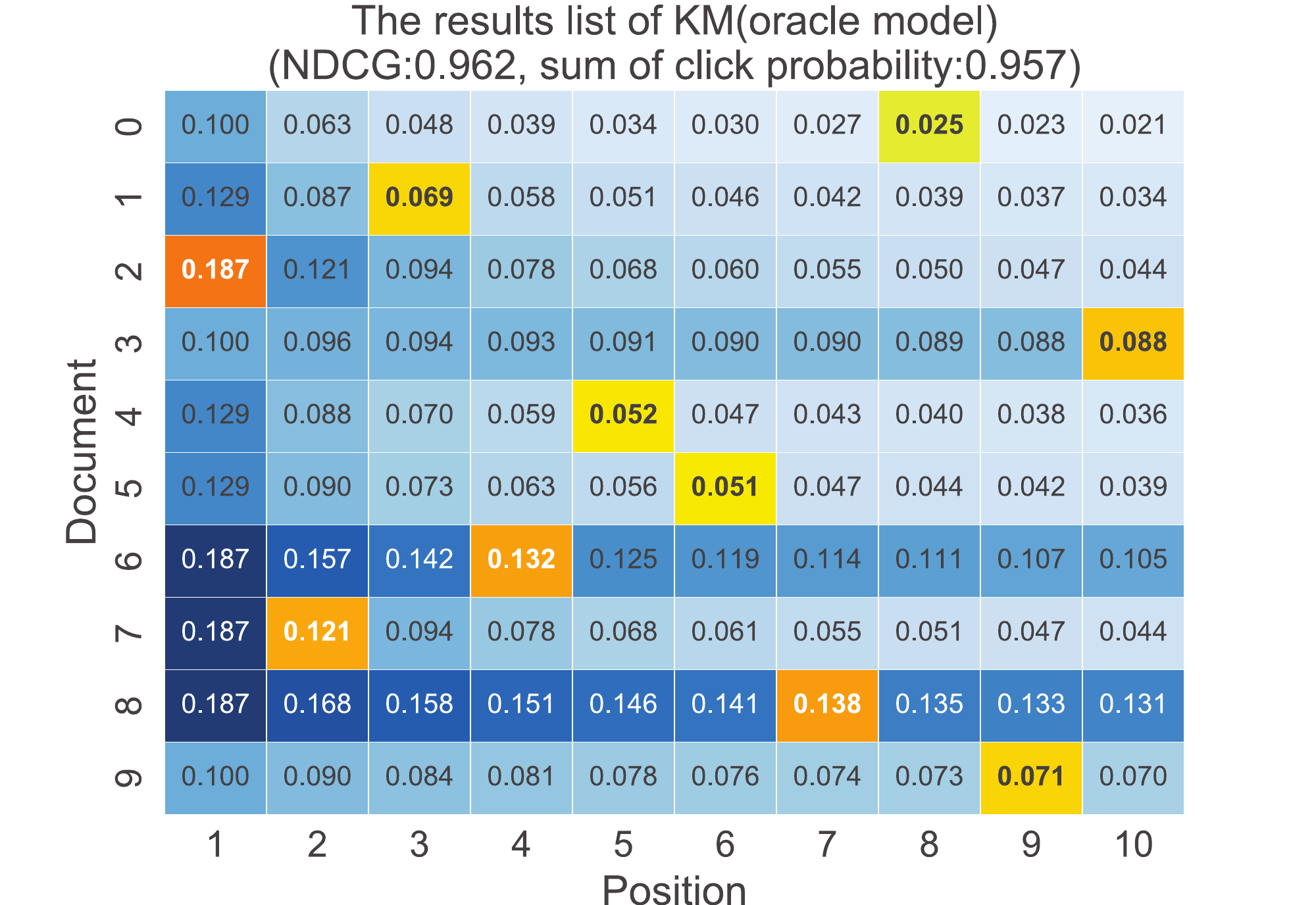}
    \includegraphics[width=0.3\textwidth]{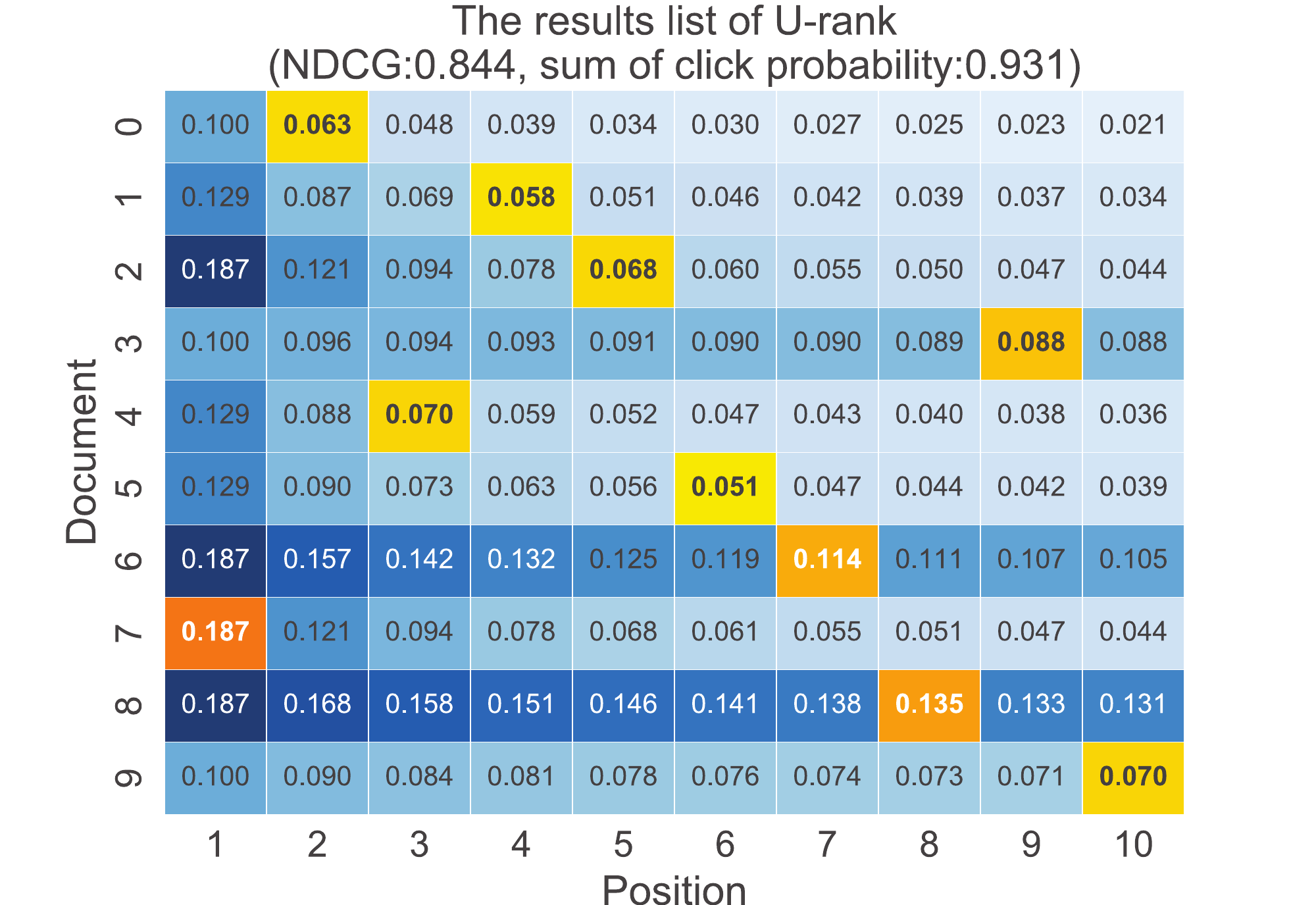}
    \includegraphics[width=0.3\textwidth]{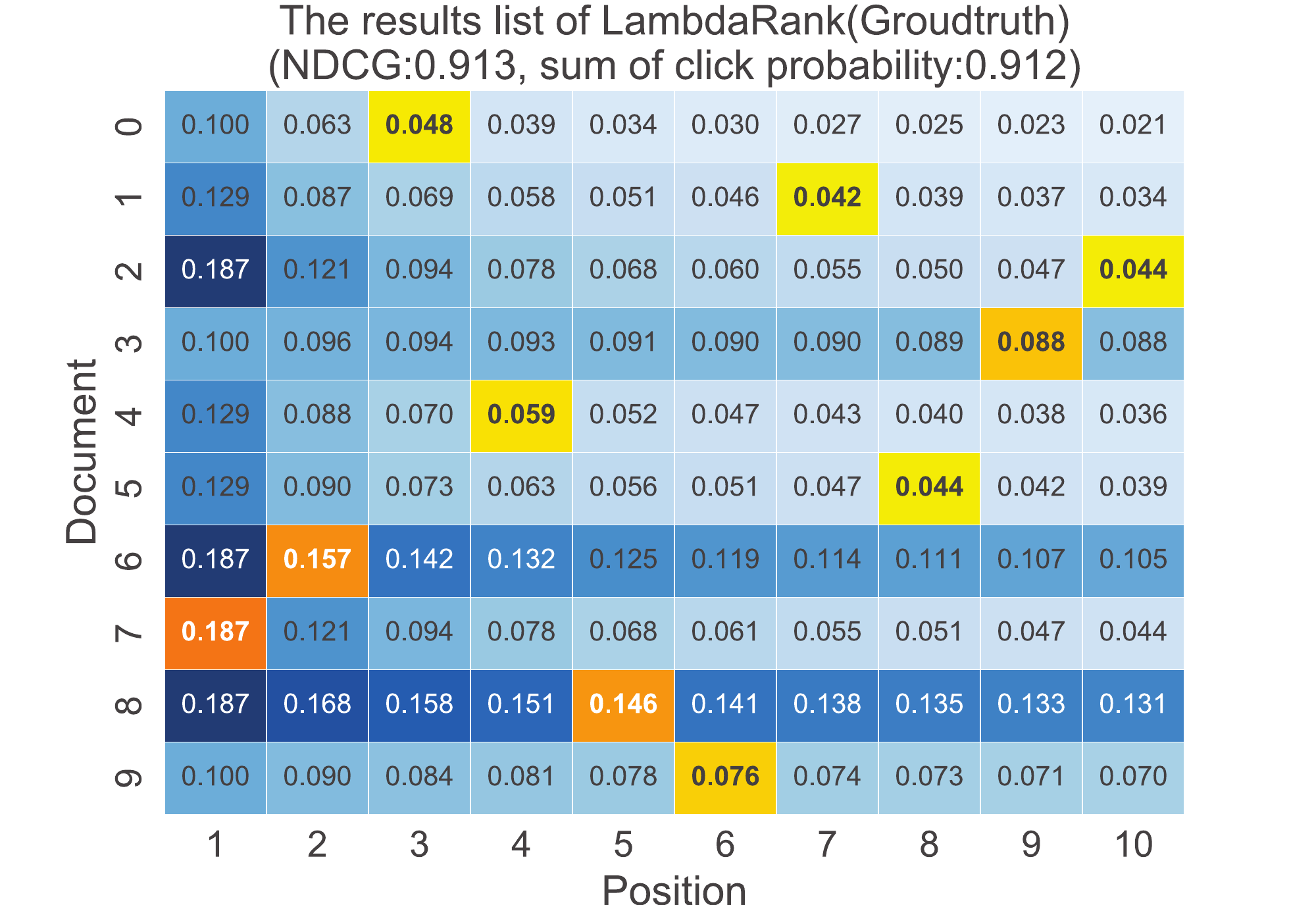}
    \vspace{-7pt}
    \caption{Comparison of the result lists of different methods on the first query of MSLR-WEB10K.}
    \vspace{-10pt}
    \label{fig:result_list}
\end{figure*}

\vspace{-3pt}
\subsection{Baselines}
We implement eight baselines that explore the performance of two standard learning to rank methods (i.e., \textbf{SVMRank}~\cite{svmrank} and \textbf{LambdaRank}~\cite{burges2007learning}), with four propensity estimation methods,
which are detailed as follows. 
(1) \textbf{None} uses the original click data without debiasing. 
(2) \textbf{Randomization}~\cite{Joachims2016UnbiasedLW} estimates propensity with online randomized experiments.
(3) \textbf{CPBM}~\cite{Fang2018InterventionHF} estimates examination probability w.r.t different queries based on intervention harvesting.
(4) \textbf{Groundtruth} uses the groundtruth examination probability for \textit{oracle model} as propensity. The result of this method is the upper bound of the results of other IPS approaches based on the same ranking model. Other methods we implement include 
(5) \textbf{CTR-1}, the position-aware click model used in our framework which assigns position 1 to each item during online inference.
(6) \textbf{DLA}~\cite{10.1145/3209978.3209986} is a dual learning algorithm that jointly learns an unbiased ranker and an unbiased propensity model.
We also explore the performance of \textbf{KM (oracle model)} which solves the maximum-weight graph matching problem via Kuhn-Munkres (KM) algorithm~\cite{kuhn1955hungarian,munkres1957algorithms}, given the groundtruth CTR. It is supposed to produce the best utility we can achieve on the test data.

\vspace{-3pt}

\subsection{Overall Performance}\label{sec:semiresults}
In this section, we assume the utility value of each item to be 1 in order to consistently and fairly compare \textit{U-rank} with existing (unbiased) learning to rank methods.
We evaluate the performance of the baseline approaches and \textit{U-rank} in terms of the relevance based metrics, i.e., \textit{MAP} and \textit{nDCG} (\textit{nDCG} denotes \textit{nDCG@10}), and utility based metrics, i.e., \textit{\# Click} and \textit{CTR}.
Here, The \textit{\# Click} and \textit{CTR} are utility metrics based on oracle click model denoting \textit{clicks per query} and \textit{CTR per document}, respectively. 
The overall performance on the three benchmark datasets is shown in Table~\ref{tab:simulation}.
    \textit{Firstly,} our method \textit{U-rank} achieves consistently the best performance over the state-of-the-art baseline approaches on the utility-based metrics, i.e., \textit{\#Click} and \textit{CTR}. For example, \textit{U-rank} achieves 1\% improvement in Yahoo LETOR set 1 and 8.3\% improvement in MSLR-WEB10K on \textit{CTR} comparing to the best baseline methods~\footnote{The baseline in italic use the information from oracle click model, so they are not included for comparison.}.
    \textit{Secondly,} 
    \textit{U-rank} also outperforms most of the baselines in terms of the relevance based metrics, i.e., \textit{MAP} and \textit{nDCG}, though it does not always perform the best especially on MSLR-WEB10K dataset where the disagreement between utility-based metric and relevance-based metric is larger than that on the other two datasets.
    \textit{Thirdly,} the method \textit{Groundtruth} achieves the best utility among the counterfactual learning approaches, demonstrating the effectiveness of the IPS-based framework when the propensity estimation is accurate. 
    Randomization fails to perform well because it assumes that the user's attention only relates to the position which is not true in our setting where the user's attention relies on both the position and the item feature.
    CPBM achieves the second-best utility among the IPS-based methods since it models the propensity of each query by taking query features into consideration. 
    \textit{Lastly,} \textit{U-rank} and CTR-1 share the same click model. However, \textit{U-rank} outperforms CTR-1 mainly because CTR-1 ranks items by their estimated CTR at position 1, which is suboptimal in case of item-specific attention bias.
    \textit{U-rank} also outperforms DLA since DLA relies heavily on the accuracy of the estimated propensity, which is hard to achieve.

\subsection{Empirical Analysis}
\noindent
\textbf{RQ1: How does our model achieve higher CTR?}
In Figure~\ref{fig:click_dist}, we show the average CTR on each position of \textit{U-rank} and LambdaRank(Groundtruth) , the upper bound of counterfactual learning methods in Table~\ref{tab:simulation}.
We also plot the results of \textit{KM (oracle model)} for reference.

\vspace{-5pt}
\begin{figure}[hbt]
	\centering
	\includegraphics[width=0.235\textwidth]{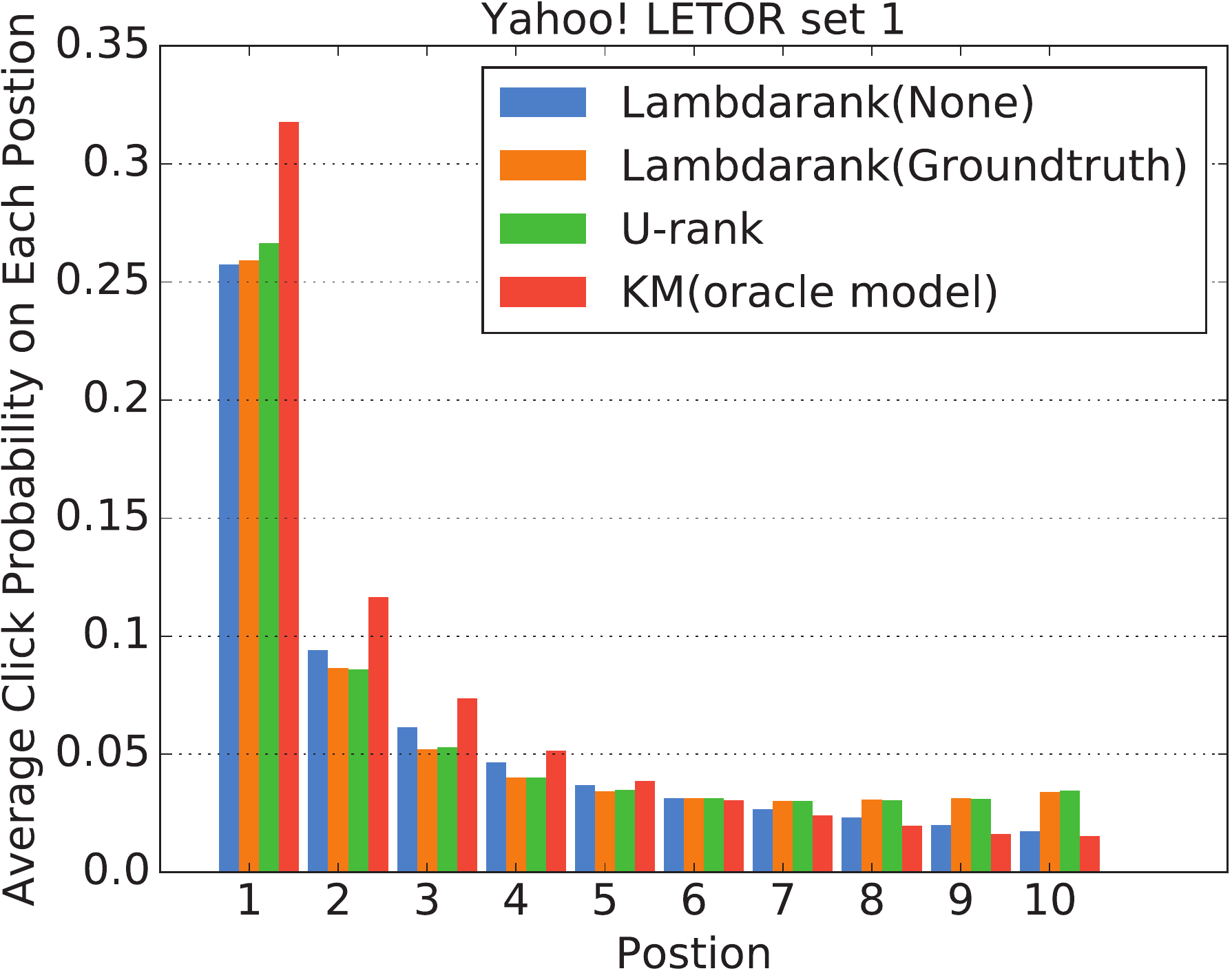}
	\includegraphics[width=0.235\textwidth]{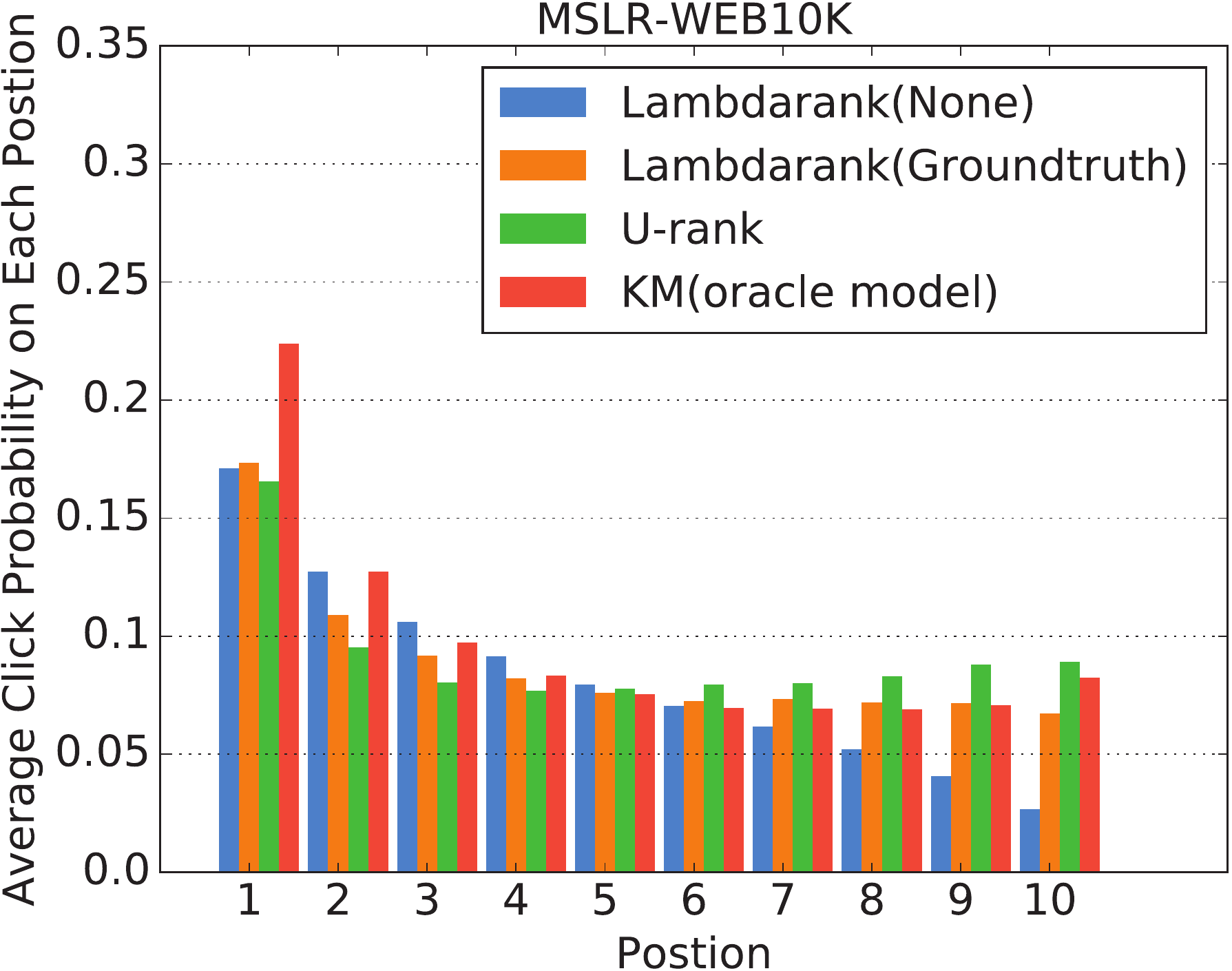}
	\vspace{-18pt}
	\caption{Average CTR on each position.}
	\label{fig:click_dist}
	\vspace{-12pt}
\end{figure}

Firstly, comparing the results of the two datasets in Figure~\ref{fig:click_dist}, we observe a steeper decline of average CTR to positions of the KM (oracle model) method on Yahoo dataset than that on the MSLR dataset. This suggests that on this dataset, positions have a very strong impact on users' click. Thus, to optimize the utility, a well-performed approach should put more relevant items at higher positions.
In MSLR-WEB10K, on the other hand, the average CTR of the optimal matching tends to be equally distributed on the positions, compared to the Yahoo dataset. 
We find that \emph{our method is adaptive to different severity of position bias}. In Yahoo dataset, our model focuses more on top positions than LambdaRank, while in the MSLR dataset, our model learns a flatter distribution. Notably, on both datasets, our model achieves a larger sum of click probabilities over all the positions than LambdaRank. 

Secondly, we analyze the result of a single query in detail. 
The experiment is conducted on the first query of the MSLR-WEB10K dataset.
Figure~\ref{fig:result_list} shows the click probabilities of the ten items for this query and their click probabilities if placed at each position according to our oracle click data generation model. The position of each item assigned by different methods is denoted in orange color. 
We can see that although LambdaRank performs better in nDCG with a groundtruth propensity. It, however, achieves a lower CTR than our method \textit{U-rank}. This is because, similar to the KM (oracle model), \emph{\textit{U-rank} will take the position sensitivity of different items into consideration}. For example, document 6 is of high relevance and relatively not sensitive to the position change. LambdaRank displays it at the second position while our method and KM both display it at a lower position, so that the second position is kept for an item that is more sensitive to the position change.

\vspace{-0pt}
\begin{table}[t]
	\caption{Comparison of two different architecture of the position aware click estimation.}
	\vspace{-10pt}
	\resizebox{0.3\textwidth}{!}{%
		\begin{tabular}{@{}lccccc@{}}
			\toprule
			& \multicolumn{3}{c}{AUC} & \multirow{2}{*}{\begin{tabular}[c]{@{}|c@{}} \# click\end{tabular}} & \multirow{2}{*}{CTR} \\
			\cmidrule(r){1-4}
			& train & test & validation &  &  \\
			\hline
			A1 & 0.695 & 0.684 & 0.687 & 0.903 & 0.915 \\
			A2 & 0.701 & 0.693 & 0.692 & 0.852 & 0.847 \\ \bottomrule
		\end{tabular}%
	}
	\vspace{-10pt}
	\label{tab:click}
\end{table}

\noindent\textbf{RQ2: What kind of architecture should we use to implement the position-aware click estimation?}
We implement two kinds of architecture for the position-aware click estimation. A1 is a neural network, with the item features as input and a $K$-dim vector as output, where the $k$-th dimension denotes the CTR of the item at position $k$, and $K$ denotes the number of positions. 
A2 is also a neural network, with the concatenation of item feature and position in one-hot encoding as its input and a single value as output, representing the CTR of the item at the given position. The result on MLSR-WEB10K is presented in Table~\ref{tab:click}. Although A2 achieves better AUC, we utilize A1 as the click model to pursue higher utility.

\vspace{-5 pt}
\section{Real-world Deployment}
\begin{table*}[t]
	\centering
	\small
	\caption{Comparison of different unbiased learning to rank models on real-world recommendation scenarios.}
	\vspace{-8pt}
	\begin{threeparttable}
	\begin{tabular}{ c| c || c | c | c | c | c | c | c | c } 
		\hline
	    \multicolumn{2}{c||}{\multirow{2}{*}{Ranking model}} & \multicolumn{4}{c |}{Scenario 1 (without bid)} & \multicolumn{4}{c}{Scenario 2 (with bid)} \\ \cline{3-10}
		\multicolumn{2}{c||}{} & \# click & \# click@1 & \# click@3 & \# click@5 & Revenue & Revenue@1 & Revenue@3 & Revenue@5  \\\hline
		\multirow{2}{*}{SVMRank}
		& None & 1.586 & 0.500 & 1.107 & 1.348 & 3.602 & 0.959 & 2.177 & 2.788  \\\cline{2-2}
 		& \textit{Groundtruth} & 1.617 & 0.536 & 1.154 & 1.386 & 3.619 & 1.015 & 2.229 & 2.825 \\ \hline
 		
 		\multirow{2}{*}{LambdaRank}
 		& None & 1.750 & 0.701 & 1.327 & 1.556 & 3.586 & 0.964 & 2.178 & 2.774  \\ \cline{2-2}
 		& \textit{Groundtruth} & 1.826 & 0.781 & 1.429 & 1.640 & 3.637 & 1.009 & 2.245 & 2.837  \\ \hline
 		\multicolumn{2}{c||}{DLA} & 1.665 & 0.624 & 1.338 & 1.520 & 3.631 & 0.958 & 2.233 & 2.827  \\ \hline
 		
 		\multicolumn{2}{c||}{DeepFM} & 1.790 & 0.762 & 1.379 & 1.593 & 3.753 & 1.131 & 2.289 & 2.881  \\ \hline
		\multicolumn{2}{c||}{U-rank} & \textbf{1.859*} & \textbf{0.841*} & \textbf{1.474*} & \textbf{1.676*} & \textbf{3.966*} & \textbf{1.264*} & \textbf{2.607*} & \textbf{3.214*} \\ \hline
		\end{tabular}
	\begin{tablenotes}
		 \item \footnotesize  $*$ denotes statistically significant improvement (measured by t-test with p-value$<$0.05) over all baselines.
	\end{tablenotes}
	\end{threeparttable}	
	\vspace{-8pt}
	\label{tab:huawei_offline}
\end{table*}

In order to verify the effectiveness of our proposed model in real-world applications, we conduct experiments in two recommendation scenarios in Company X's App Store. This App Store has hundreds of millions of daily active users who create hundreds of billions of user logs every day in the form of implicit feedback such as browsing, clicking, and downloading behaviors.

\subsection{Offline Evaluation}
\textbf{Setups.}
We conduct offline experiments based on two recommendation scenarios with different utility settings.
In Scenario 1, we only consider the downloads of the Apps as the utility, while in Scenario 2,  the bid price of each App download needs to be considered.
In both scenarios, we use seven consecutive days' data for training and the eighth day's data for testing. As in the semi-synthetic experiments, we also implement two LETOR methods, \textit{i.e.}, SVMrank and LambdaRank as baselines. The propensity estimation method \textit{Randomization} is not applicable here since we are not allowed to randomly swap two items of a ranked list in a live recommender system. 
Similarly, CPBM is not applicable either since in practice we cannot obtain the ranking results of the same user from multiple rankers at the same time.
Thus, we only compare \textit{U-rank} with the ranker learned with biased click data, i.e., \textit{None}, and the ranker with groundtruth propensity, \textit{i.e.}, \textit{Groundtruth}. 
The groundtruth propensity is the same as the propensity that we use in evaluation in the next paragraph.
\textit{DeepFM} is included as a baseline as it is the production baseline in this App recommendation online system. It trains with position feature and takes default position 1 in the inference stage. 
To make a fair comparison, we perform DeepFM architecture in both click model and ranking model in \textit{U-rank}.

\noindent\textbf{Metrics.} Unlike in the semi-synthetic experiments, here we do not know the underlying user click model. 
Thus, we have to debias the click data generated by a historical ranker with a pre-estimated propensity to obtain the click signals on the new positions.
We estimate the propensity for each category of items from 120 days' click data on random traffic. This category-wise propensity estimation is a coarse approximation of the groundtruth propensity, which is not available. The propensity is defined as $Q(o_i=1|\text{category}_i, k_i)$, where $\text{category}_i$ denotes the category of item $i$. This propensity is only used for evaluation except in the $Groundtruth$ methods in Table~\ref{tab:huawei_offline}, where this propensity is used for debiasing.

The utility in Scenario 1 is defined as the expected number of debiased clicks at the top-$K$ positions in a session, i.e., $\# click@K = \sum_{k_i=1}^{K} \frac{Q(o_i=1|\text{category}_i, k_i)}{Q(o_i=1|\text{category}_i, k_i^h)} c_{i, k_i^h}$. We use $\# click$ to denote the case when $K=n_q$.
The utility in Scenario 2 is defined as the expected revenue at top-$K$ positions in a session after debasing, i.e., $revenue@K = \sum_{k_i=1}^{K} \frac{Q(o_i=1|\text{category}_i, k_i)}{Q(o_i=1|\text{category}_i, k_i^h)} c_{i, k_i^h} \cdot b_i$ where $b_i$ is the bid of item $i$. We use $Revenue$ to denote the case when $K=n_q$.

\noindent\textbf{Results.} The overall performance on the two real-world datasets is shown in Table~\ref{tab:huawei_offline}. We have the following observations.
\textit{Firstly}, \textit{U-rank} achieves the best performance comparing to the state-of-the-art baselines. Specifically, in Scenario 1, \textit{U-rank} achieves 1.8\%, 7.7\%, 3.1\% and 2.2\% improvement over the best baseline method in terms of \#click, \#click@1, \#click@3 and \#click@5, respectively. In the experiment with bid in Scenario 2, the improvements are 2.5\%, 13\%, 6.7\% and 3.9\% in terms of Revenue, Revenue@1, Revenue@3 and Revenue@5, respectively. These results demonstrate the superiority of our approach over the baselines in optimizing the utility, which motivates us to deploy \textit{U-rank} in the live recommender system. 
\textit{Secondly}, \textit{U-rank} performs better than DeepFM because \textit{U-rank} considers item-specific attention bias while DeepFM learns from biased data.
We do not elaborate on the other results since they are consistent with the results in the semi-synthetic experiments.

\begin{figure}
    \centering
    \includegraphics[width=0.48\textwidth]{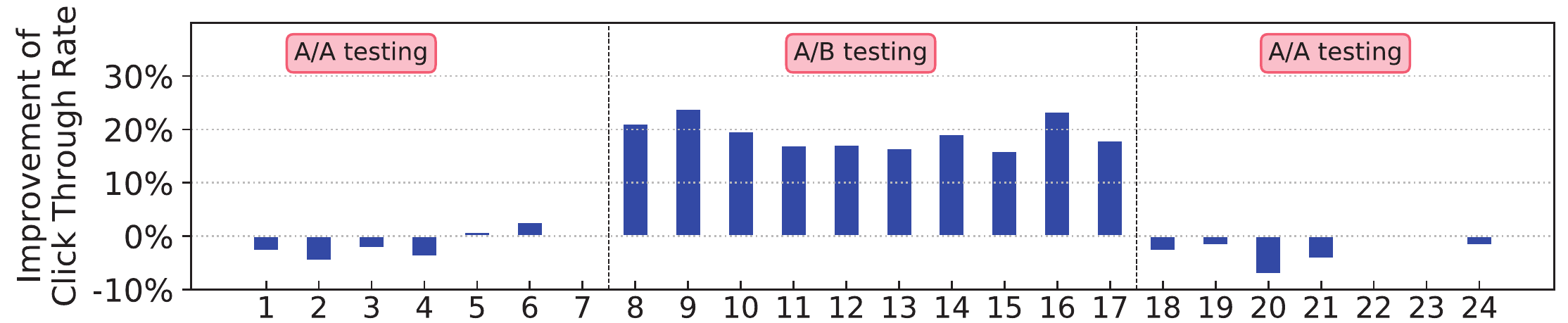}
    \vspace{-8pt}
    \caption{Online experimental results of click through rate.}
    \label{fig:res_live_ctr}
    \vspace{8pt}
    \includegraphics[width=0.48\textwidth]{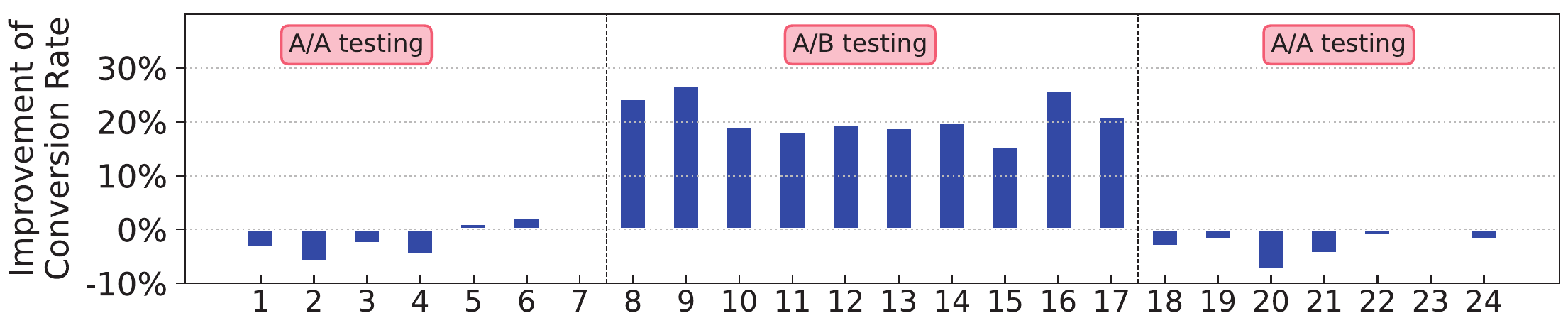}
    \vspace{-8pt}
    \caption{Online experimental results of conversion rate.}
    \vspace{-25pt}
    \label{fig:res_live_cvr}
\end{figure}

\subsection{Online Evaluation}
\textbf{Setups.}
We conduct A/B testing in a recommendation scenario in Company X's App store, comparing the proposed model \textit{U-rank} with the current production baseline  DeepFM~\cite{DBLP:conf/ijcai/GuoTYLH17} that supports multiple scenarios such as ``Must-have Apps'' and ``Novel and Fun''.
The whole online experiment lasts 24 days, from May 6, 2020 to May 29, 2020.
We monitor the results of A/A testing for the first seven days, conduct A/B testing for the following ten days, and conduct A/A testing again in the last seven days.
15\% of the users are randomly selected as the experimental group and another 15\% of the users are in the control group. During A/A testing, all the users are served by DeepFM model \cite{DBLP:conf/ijcai/GuoTYLH17}.
During A/B testing, users in the control group are presented with recommendation by DeepFM, while users in the experimental group are presented with the recommendation by our proposed model \textit{U-rank}. 
Note that the click model of \textit{U-rank} shares the same network architecture and parameter complexity with DeepFM in order to verify whether the improvement is brought by the objective function design of the ranker in \textit{U-rank}. 

To deploy \textit{U-rank}, we utilize a single node with 48 core Intel Xeon CPU E5-2670 (2.30 GHZ), 400 GB RAM and as well as 2 NVIDIA TESLA V100 GPU cards, which is the same as the training environment of the baseline DeepFM. For model training, \textit{U-rank} requires minor changes to the current training procedure due to the pair-wise loss function. For model inference, \textit{U-rank} shares the same pipeline as DeepFM, which means there is no extra engineering work needed in model inference, to upgrade DeepFM model (or other similar deep models) to \textit{U-rank}.

\noindent\textbf{Metrics.}
We examine two metrics in the online evaluation. 
They are \textit{Click-through rate}: $CTR = \frac{\# downloads}{\# impressions}$  and \textit{Conversion rate}: $CVR=\frac{\# downloads}{\# users}$, where \# downloads, \# impressions and \#users are the number of downloads, impressions and visited users, respectively.

\noindent\textbf{Results.}
Figure~\ref{fig:res_live_ctr} and Figure~\ref{fig:res_live_cvr} show the improvement of the experimental group over the control group with respective to CTR and CVR, respectively.
We can see that the system is rather stable where both CTR and CVR fluctuated within 8\% during the A/A testing.
Our \textit{U-rank} model is launched to the live system on Day 8.
From Day 8, we observe a significant improvement over the baseline model with respect to both CTR and CVR.
The average improvement of CTR is 19.2\% and the average improvement of CVR is 20.8\% over the ten days of A/B testing.
These results clearly demonstrate the high effectiveness of our proposed model in improving the total utility which refers to the number of downloads in this scenario.
From Day 18, we conduct A/A testing again to replace our \textit{U-rank} model with the baseline model in the experimental group.
We observe a sharp drop in the performance of the experimental group, which once more verify that the improvement of online performance in the experimental group is indeed introduced by our proposed model.

\vspace{-5pt}
\section{Conclusion}
In this paper, we propose a novel framework \textit{U-rank}, which directly optimizes the expected utility of the ranked list without any extra assumption on relevance nor examination. 
Specifically, \textit{U-rank} first uses a position-aware deep CTR model to perform an unbiased estimation of the expected utility, and then optimizes the objective with an efficient algorithm based on a LambdaRank-like objective.
Extensive studies on three benchmark datasets and two real-world datasets based on different scenarios have shown the effectiveness of our work. We also deploy this ranking framework on a commercial recommender system and observe a large utility improvement over the production baseline via online A/B testing. In future work, we plan to consider other biases like selection bias and propose a more general debiasing framework. 

\vspace{-5pt}
\section*{Acknowledgement}
    The corresponding author Weinan Zhang thanks the support of NSFC (61702327, 61772333, 61632017). The work is also sponsored by Huawei Innovation Research Program.
	\bibliographystyle{ACM-Reference-Format}
	\bibliography{sample-base}

\end{document}